\documentclass[journal, 12pt, onecolumn, draftclsnofoot,]{IEEEtran}
\input{common-packages}
\usepackage{subcaption}
\usepackage{iftex}
\ifLuaTeX
\usepackage{fontspec}
\else\fi
\ifXeTeX
\usepackage{fontspec}
\else\fi
\ifPDFTeX
\usepackage[utf8]{inputenc}
\else\fi

\pgfplotsset{compat=newest} 
\usepgfplotslibrary{units} 
\bibliography{mINFOCOM18, massivemimo}
\begin{document}
 \author{
   \IEEEauthorblockN{
   Shuai~Zhang,~\IEEEmembership{Student~Member,~IEEE},
   Lu~Liu,~\IEEEmembership{Student~Member,~IEEE},
   Yu~Cheng,~\IEEEmembership{Senior~Member,~IEEE},
   Xianghui~Cao,~\IEEEmembership{Senior~Member,~IEEE},
   Sheng~Zhou,~\IEEEmembership{Member,~IEEE},
   Zhisheng~Niu,~\IEEEmembership{Fellow,~IEEE},
   Hangguan~Shan,~\IEEEmembership{Member,~IEEE}}%
   \thanks{S. Zhang, L. Liu, Y. Cheng are with Department of Electrical and Computer Engineering, Illinois Institute of Technology, Chicago, IL 60616 (email: \{szhang104,lliu\}@hawk.iit.edu; ycheng@iit.edu)}
   \thanks{X. Cao is with School of Automation, Southeast University, Nanjing, 210018, China (email: xh.cao@ieee.org)}
   \thanks{S. Zhou and Z. Niu are with Tsinghua National Laboratory for Information Science and Technology, Department of Electronic Engineering, Tsinghua University, Beijing 100084, China (email: \{sheng.zhou, niuzhs\}@tsinghua.edu.cn)}
   \thanks{H. Shan is with the College of Information Science and Electronic Engineering, Zhejiang University, Hangzhou 310027, China (e-mail: hshan@zju.edu.cn)}}
\newcommand{\citeneeded}[1][]{\textsuperscript{\color{blue} [citation needed: #1]}}
\title{Energy Efficient Massive MIMO through Distributed Precoder Design}
\maketitle
\begin{abstract}
This paper presents an energy-efficient downlink precoding scheme with the objective of maximizing system energy efficiency in a multi-cell massive MIMO system. The proposed precoding design jointly considers the issues of power control, interference management, antenna switching and user throughput in a cluster of base stations (BS). We demonstrate that the precoding design can be formulated into a general sparsity-inducing non-convex problem, which is NP-hard. We thus apply a smooth approximation of zero-norm in the antenna power management to enable the application of the gradient-based algorithms.
The non-convexity of the problem may also cause slow convergence or even divergence if some classical gradient algorithms are directly applied. We thus develop an efficient alternative algorithm combining features from augmented multiplier (AM) and quadratic programming (QP) to guarantee the convergence. We theoretically prove the convergence conditions for our algorithm both locally and globally under realistic assumptions. Our proposed algorithms further facilitate a distributed implementation that reduces backhaul overhead by offloading data-intensive steps to local computation at each BS. Numerical results confirm that our methods indeed achieve higher energy efficiency with superior convergence rate, compared to some well-known existing methods.
\end{abstract}


\section{Introduction}
A central theme in developing the next generation mobile communication system is the drive for higher data rate at a lower energy consumption. Today's communication system designer must be aware of both the throughput performance as well as the corresponding power cost, and as a result Energy Efficiency (EE) has become an important metric for network evaluation and optimization \cite{feng_harvest_2016,miao2010energy,cao_towards_2018}. EE is usually defined as the ratio between the achieved throughput (in bits/s) and the corresponding power consumption. Since in today's communication systems there are increasingly more interacting subsystems, it is not immediately obvious whether previously proposed network control schemes, aiming to achieve either high system throughput or low power, can still work under high EE requirements. It is a worthy question to see if it is viable to follow the network utility optimization path to design a network: how to maximize EE when the scenario is constrained by quality-of-service (QoS) requirements and power limits, or maintaining a certain level of fairness among users and base stations (BSs) as reported in \cite{chen_adaptive_2016,nguyen_distributed_2017,zappone_globally_2017}.

For achieving new levels of high EE communication, Massive MIMO is a promising technology currently under test and development for the 5G communication standard. Its salient feature is that BSs are equipped with an excessive number of antennas, much higher than the number of served users, to achieve aggressive diversity gain. The enabling mechanism behind this benefit is \emph{favorable propagation}, i.e., when the number of BS antennas $M \to \infty$, the channel characteristics become almost deterministic, and the radio links become near-orthogonal to each other such that the effects of fading, intra-cell interference and uncorrelated noise disappear. Under an appropriate system configuration, ideally very large multiplexing and array gains can be achieved. According to the theoretical analysis in \cite {marzetta_noncooperative_2010} \cite{hoydis_massive_2013}, a system with $M$ antennas and $K$ users ($M \gg K$) could achieve $O(M)$ gain in effective signal-to-noise-ratio (SNR), indicating achievable equal throughput with only $1/M$ power consumption compared to the single-antenna system. On the other hand, this technology would necessarily require deploying many antennas and the corresponding radio circuits, which are known to have low power-efficiency limitation\cite{liu_energy_2017}. This might be a source of power waste when the traffic demand is low. In sum, Massive MIMO has great potential in improving EE in cellular systems, but would require careful planning and design choice to achieve it.

The gain in EE would even be greater if Massive MIMO is combined with small cells deployment. As shown in theoretical analysis\cite{rajoria_comprehensive_2018}, smaller cell size and higher BS antenna count can both contribute to higher throughput. A typical example of it is Heterogeneous Networks (HetNet): a central master base station (MBS) provides coverage for a \emph{macro cell}, and within it many small base stations (SBSs) form their own \emph{small cells} as an overlay. Such an architecture is especially useful in densely populated urban areas, where a large part of the traffic comes from confined areas with high traffic called \emph{hotspots}. A central BS may not be able to provide satisfactory QoS in an energy-efficient manner, due to either congestion, interference or unfavorable channels. In this case, small cells not only offload part of the traffic, but also reduce the power requirements for MBS because of geometrical proximity to the users. In order to take full advantage of such a heterogeneous setup for high QoS provisioning or energy efficiency, it is important to coordinate the BSs and optimize the resource allocation in HetNets \cite{adhikary2015massive,tang2015resource,cai2016green}. For example, when the SBSs make use of the same frequency resource, the inter-tier interference can be significant; due to their small sizes users in small cells are more likely to suffer from inter-cell interference (ICI).\cite{grieger2015multicell,dahrouj_coordinated_2010}

Precoding is the technique for enabling multi-user transmission in a multi-cell cluster. By varying the amplitude and phase of different antennas, linear precoding schemes like Zero-Forcing or minimum mean-square error could achieve near-optimal capacity, and they help suppress interference with available channel multiplicity. However, it is not clear if in the context of energy-efficiency a dynamically optimized precoding scheme can do better. This is especially important in large-scale MIMO systems with small cell deployment, where base stations may be equipped with tens or hundreds of antennas, since inter-cell interference and antenna power management may become dominant factors for energy-efficiency considerations. 
Based on the above points, we study the precoder design problem maximizing the EE, in the scenario of a cooperating Massive MIMO-enabled small cell BS cluster. We propose a framework which jointly considers factors including power control, BS antenna switching and interference, conventionally solved as separate problems. Towards this end we show that an efficient and distributed algorithm can be applied to solve the problem.

The main contributions of this paper can be summarized as the following:
\begin{itemize}
\item We consider the downlink transmission of multiple massive MIMO BSs, particularly the SBSs in a cooperative cluster. They jointly determine their power control, precoding vectors and BS antenna switching to achieve high system EE. We show that it is a general sparsity-inducing non-convex problem with a separable structure, which is considered an NP-hard problem to solve for a global solution.

\item We leverage the separable structure to transform the problem so that it can be solved in a distributed manner. The smooth approximation of zero-norm constraints are given, and we also show the series of transformations needed to arrive at the iterative algorithm combining features from Newton's Method and augmented multiplier method. We first give a theoretical proof to show that it is convergent when the initial solution is sufficiently close, which is not a trivial property when the objective is non-convex. We then show the additional steps it requires to guarantee convergence for any initial solution.

\item We present numerical results to demonstrate the effectiveness of the proposed algorithm, both in terms of convergence speed and achieved energy efficiencies, under different system parameters, with results from other schemes for comparison. Some useful conclusions for system deployment can be learned from the results.
\end{itemize}

\textit{Organization and Notation:} The remainder of the paper is organized as follows. Section \ref{sec-related} introduces more current state of this area of research. Section \ref{sec-model} states the system model and problem formulation. In Section \ref{sec-anal} the motivation, derivation and analysis of the algorithm and the transformation of the problem needed are given. Section \ref{sec-num} presents the numerical results for performance evaluation and comparison with other methods, and the lessons learned from the application of such method. Section \ref{sec-summary} concludes the work and gives additional remarks. Mathematical notation note: in this paper we use calligraphic letters ($\mathcal{A}$) for sets, capital letters ($A$) for the set cardinality and corresponding lower-case letter as a specific member of the set ($a \in \mathcal{A}$). Bold letters denote a vector or matrix, and square brackets with subscript is a vector by enumeration, e.g. $[x_{b}]_{b\in \mathcal{B}}$ is a set of all $x_b$ when $b$ ranges from the set $\mathcal {B}$. $\circ$ denotes element-wise product.

\section{Related Work}\label{sec-related}
The precoding process refers to the scaling and phase changing manipulation of transmitting signals such that the received signals could have desired properties. Massive MIMO is a natural extension of the Multi-User MIMO, where the massive number of antennas can only be put to use when precoding vectors are properly selected for a system design goal. The original paper \cite{marzetta_noncooperative_2010} gives a closed-form expression to the key system performance metric and discusses the choice of key system parameters like antenna and user number. Based on those, there have been many works on the design guidelines to optimize system performance metrics \cite{bjornson_massive_2016,jiang_joint_2018}. However, those results assume certain kinds of precoding schemes and the problem of designing energy-efficient system when precoding scheme is also considered has yet to receive enough attention. The mainstream approach is linear precoding, which calculates the coefficients for linear combination at the receiver, for example zero-forcing, and signal-to-leakage-and-noise-ratio(SLNR). They are cheap to implement at the cost of sub-optimal system throughput. One thread of research is to extend these results to multi-cell scenarios and distribute the computation; however many do not consider the energy implications and could be operating in low EE regime. Also a central controller who are assumed to have CSI information from the BSs could bring high communication and processing overhead which could result in hidden energy costs; another related, but different approach is to maximize the minimum SINR with power constraints. While easier to solve, these problems often share the drawbacks of not adapting to the current traffic: the constraints need to be reset and found by hand or another process in order to operate efficiently.


For advanced convex optimization techniques in the multi-agent setting, currently the most popular method is Alternating Direction Multiplier Method. In \cite{shen_distributed_2012} ADMM is used in combination with Semidefinite Relaxation to solve the coordinated beamforming with uncertainties modeled as ellipsoid in the CSI for throughput gain. The authors of \cite{wang_robust_2013} further extend the work to consider a general form of CSI uncertainties that gives closed-form solutions in the strongly convex cases. However they invariably need to rely on the usage of Semidefinite relaxation, which ignores one of the matrix rank constraints in order to readily apply ADMM. A joint solution regarding BS clustering and beamforming is given in \cite{hong_joint_2013}. \cite{chang2014distributed} provides a general framework for using primal-dual perturbation method, which can be seen as a version of the multiplier methods for optimization that has involved and coupled constraints. Another good introduction of the ADMM algorithm use in wireless network is presented in \cite{shi_large-scale_2015}, which gives the result of how to obtain infeasibility certificate and speed-up tricks. Yet they all have to either model the problem in the convex form, which is limiting for EE design goal, or use convex approximation at local iterations, incurring additional complexity. Another important mathematical tool to deal with CSI uncertainty is random matrix theory, where the linear precoders could be adapted to the stochastic forms according to the level of available channel knowledge. While providing a low-cost computation with reasonable performance, the drawback in this approach is that they need to be built on existing assumptions on the modeled system, e.g., the used precoding scheme, knowledge on the channel state, making their results dependent on the system specification and lacking the ability to generalize.

\section{System Model and Formulation}\label{sec-model}
\subsection{System Model}
We consider a Massive MIMO network deployed in a two-tier HetNet topology, with a macro-cell covered by MBS and small cells of SBSs on top of it, as illustrated by Fig. \ref{fig-system}. There is a set of BSs $\mathcal{B} = \{0,1,\cdots, B\}$, with MBS as the $0$-th one, and the other $B$ SBSs. Each BS $b \in \mathcal{B}$ is equipped with a set of antennas $\mathcal{N}_{b}$, which could be different as the network may be heterogeneous. Each of the $K$ single-antenna users $u_k \in \mathcal{U}$ can associate with one or more BSs and suppose that the set of users associated with BS $b$ is $\mathcal{U}_{b}$. For clarity, the associated BSs of a given user $u_k$ is given by the set $\mathcal{B}_k$.
\begin{figure}[h]
  \centering
  \includegraphics[width=0.5\textwidth]{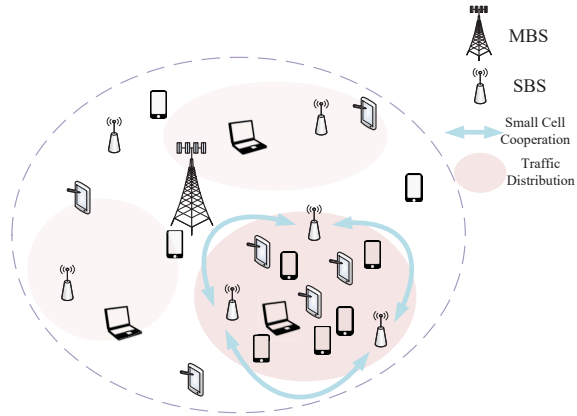}
  \caption{The system diagram of the scenario under consideration}
  \label{fig-system}
\end{figure}

The system works in the Time Division Duplex (TDD) mode, and we consider Rayleigh fading in the channels. During one coherence block, the channel is considered flat and we model the received base-band signal of user $u_k$ from base station $b$ as
\begin{align}\label{eq-dlsig}
  y_{b} = & \sqrt{p_{b,k}} \mathbf{h}^{H}_{b,k} \mathbf{w}_{b,k} s_{b,k} + \sum_{i \neq k, i\in \mathcal{U}_b} \mathbf{h}^{H}_{b,k} \mathbf{w}_{b,i} s_{b,i} \nonumber \\
            & + \sum_{i\in \mathcal{B}\setminus b}  \mathbf{h}^{H}_{i,k} \sum_{j \in \mathcal{U}_i} \mathbf{w}_{i,j} s_{i,j} + \mathbf{n}_{b}
\end{align}
, where $\mathbf{h}_{b,k}\in \mathbb{C}^{N_b}$ represents the channel between BS antennas and the single user antenna, and with the Rayleigh assumption satisfying i.i.d distribution $\mathbf{h}_{b,k} \sim \mathcal{CN}(\mathbf{0}, \beta_{b,k} \mathbf{I}_{N_b})$. The magnitude expressed as $\beta_{b,k}$ is the large-scale fading between the BS and the user, caused by path loss and other environment diffraction. 

The vectors $\mathbf{w}_{b,k} \in \mathbb{C}^{N_b}$ are the precoding vectors used by BS $b$ for user $k$. This is calculated based on the estimation of channel state information (CSI) in order to utilize the multiplicity provided by the available array of antennas when transmitting DL signals. For any BS $b$, the transmitted signal is the weighted linear combination for all users:
\begin{equation}\label{eq:ue-received}
\mathbf{x}_b = \sum_{k\in\mathcal{U}_b} \mathbf{w}_{b,k} s_{b,k}
\end{equation}

In the general case all BSs are allowed to transmit to all users, but according to the recent studies it is often more efficient to have only one BS per UE, which is the assumed operating mode here. Each term in Eq. (\ref{eq-dlsig}) is the intended signal power, intra- and inter-cell interference, and all uncorrelated signal plus receiver noise. $s_{b,k} \sim \mathcal{CN}(0,1)$ is the intended signal from BS $b$ to user $k$ and is assumed to be drawn from a Gaussian code book.

In terms of power relations, from the above assumptions we can write the Signal-to-Interference-and-Noise per user ($\text{SINR}_{bk}$) as given by
\begin{subequations}
\begin{align}
  \text{SINR}_{b_k} &= \frac{|\mathbf{h}^H_{b,bk} \mathbf{w}_{bk}|^2}{I_{bk} + N_{bk}} \label{eq:dlsinr} \\
             I_{bk} &= \sum_{i \neq k, i\in \mathcal{U}_b} |\mathbf{h}^{H}_{b,bk} \mathbf{w}_{bi}|^2 + \sum_{i\in \mathcal{B}\setminus b} \sum_{j \in \mathcal{U}_i} |\mathbf{h}^{H}_{i,bk} \mathbf{w}_{ij}|^2 \label{eq:interference}
\end{align}
\end{subequations}
where $I_{bk}$ is the interference power expressed as the sum of intra- and inter-cell interference, and $N_{bk}$ is the noise power of the user across the used spectral band. 

The rate achievable to user $b_k$ is then
\begin{equation}\label{eq:dlrateuser}
  r_{bk} = \log(1+\text{SINR}_{b_k})
\end{equation}

\subsection{Power Model}
The power of a Massive MIMO BS station should be modeled in a way that reflects the real-world operating costs. Simplified models where the amortized power level decreases to zero as the number of antennas goes to infinity is unrealistic in EE analysis, since the diminishing return of SE and non-linear increase in system power cost is not considered. To remedy this, the power consumption of a BS is divided into dynamic and static parts. The dynamic part represents the precoding controlled radio transmission power, and is a function of the throughput, which ultimately depends on factors like beamforming vectors, power control and number of active antennas, which are the control variables in response to the traffic.

One major part of the static power lies with the hardware power consumption $P_{\text{hard}}$, which consists of the active radio power, which is a linear function of the number of active antennas. This is to represent the associated radio power cost that is more of an overhead, not dependent on the activity level as opposed to the dynamic power.

Another source of power models the signal processing and backhaul transmission $P_{\text{SP}}$, which is a function of the amount of needed data for computation and transmission\cite{rubio_energy-aware_2014}, estimated from the algorithm. It incorporates the power used to encode/decode symbols, estimate channels and calculate necessary control signals in addition to the backhaul transmission costs. All radio-related power terms are further scaled by a hardware efficiency factor $\theta_b$, so as to model the effects of different hardware capabilities in the presence of heterogeneous transmitter.
In essence the base station power under consideration is the sum of the following (base station subscript $b$ omitted for clarity):
\begin{equation}
\begin{aligned}\label{eq:power}
P &= P_{\text{dyn}} + P_{\text{hard}} + P_{\text{SP}}\\
               &= \theta (\text{Tr}(\bm{U}\bm{U}^H) + N_b P_{\text{ant}}) + P_{\text{fixed}} + \sum_{k \in \mathcal{U}} r_k \cdot P_{\text{SP-Unit}}
\end{aligned}
\end{equation}
where $\bm{U}$ is the shorthand for $\mathbf{w}\mathbf{w}^T$, the correlation matrix for precoding vector; $N_b$ is the number of \emph{active} antennas with the constant $P_{\text{ant}}$ the overhead power cost per antenna; $P_{\text{fixed}}$ is the constant term for other power consumption and $r_k$ is the data rate for the $k$-th user in the current cell; $P_{\text{SP-Unit}}$ is the signal processing power per unit data flow.

\subsection{Sparse Solution}
Since the base station power constitutes almost 60\% of all power usage in a communication system\cite{sawahashi_coordinated_2010}, and the majority of them dedicated to radio circuits, it is necessary to try to find precoding vectors that utilize few antennas when it is feasible, e.g., when the QoS requirements are satisfied in low data-demand scenarios. Although this may contradict with the main features of Massive MIMO, whose gain results from adding BS antennas to maximize the utilization of degree-of-freedoms available in the channel, it could be beneficial to turn off the excess antennas when high throughput is less important than energy saving. In terms of optimization this could translate into a penalty term that measure the sparseness of the solutions. 
The optimization problem combining constraints on sparsity has been explored, e.g. in \cite{majumdar2009fast}. Optimally, antennas which when put together do not improve diversity gains should not be selected. From this starting point many reported heuristics base their calculation on the channel correlation \cite{hu2015new}, or in the simple case select those with the strongest channel. However, it remains to be seen if there is a way to directly calculate a good subset of antennae for transmitting in a real-world setting, because the problem is a combinatorial programming known to be NP-hard. This class of sparse solution problem finds many uses in classification, machine learning and signal processing. One common way of approximating this problem is to use $\mathcal{L}_1$ norm. This is a commonly used technique and they are equivalent asymptotically in the high dimension regime. In this paper we use the $\mathcal{L}_1$ norm of $\mathbf{w}$ vector for the sparseness metric.
To explore such possibilities, the $N_b$ in Eq. \ref{eq:power} is taken as the number of non-empty vectors in the precoding matrix $\bm{W}_b$:
\begin{equation}\label{eq:antenna-count}
\hat{N}_b = || \bm{W}^H \bm{W} ||_0 
\end{equation}
\newcommand{\opnzc}[1]{|| #1 ||} 
\newcommand{\anybs}{\forall b \in \mathcal{B}_C}
Such an addition is by no means trivial; it forces the solved solution to have group sparsity in the precoding vectors. $\mathcal{L}_0$ norms as constraints essentially transform the original problem to a combinatorial optimization --- since it is reducible to solving a $\texttt{Optimum Subset}$ problem, which is of NP-hard class. Moreover, its presence is troublesome for numerical algorithms as it is not differentiable, therefore to use the usual gradient-based algorithms transformations are necessary to make it tractable. 

\subsection{Problem Formulation}
The problem of maximizing the energy efficiency utility within a BS cluster $\mathcal{B}$ can be formulated as such: the optimization variables are $\{\mathbf{w}_{bk}\}_{bk}$, the set of all precoding vectors, and $\theta_b$ the individual hardware efficiency of base station $b$, and $U$ is a utility function of all BS energy efficiencies defined on $\mathbb{R}^B \mapsto \mathbb{R}$. We will use a linear combination of the individual EE for illustration, although as long as it is twice differentiable with a separate structure: $U(\bm{\eta}) = \sum_{b} U_b(\bm{\eta}_b)$ our algorithm still applies. Constraints Eq. (\ref{eq:req-user-rate}) requires that individual user throughput needs to be larger than a fixed lowest value; Constraint Eq. (\ref{eq:req-bs-rate}) specifies that each BS is limited by the available amount of backhaul. Eq. (\ref{eq:req-bs-power}) provides a hard limit on the total power per BS, denoted as $\hat{P}_{bb}$, which is determined by the maximum allowed operation power. The rest of the constraints specifies the throughput and antenna count expressions used. Then the problem is expressed as:
\begin{subequations}\label{eq:prob1}
  \begin{alignat}{2}
    \underset{\{\mathbf{w}_{bk}\}_{bk}}{\text{maximize}}~&
    & &U(\bm{\eta}) = \sum_{b\in\mathcal{B}}c_b\eta_b = \sum_{b \in \mathcal{B}} c_b \frac{\sum_{k\in\mathcal{U}_b} r_{bk}}{P_b}\\
    \text{s.t.}~&
    & & r_{bk} \geq \underline{r}_{bk} \quad \forall b \in \mathcal{B}~\forall k \in \mathcal{U}_b \label{eq:req-user-rate} \\
    &&& \sum_{k} r_{bk} \leq \bar{r}_{b} \quad \forall b \in \mathcal{B}~\forall k \in \mathcal{U}_b \label{eq:req-bs-rate} \\
    &&& \theta_b (\sum_{k \in \mathcal{U}_b} ||\mathbf{w}_{bk}||_2^2 + N_b P_{\text{ant}}) + P_{\text{fixed}} + \nonumber\\
    &&& \sum_{k \in \mathcal{U}_b} r_{bk} \cdot P_{\text{SP-Unit}} \leq \bar{P}_{bb}, \quad \forall b \in \mathcal{B}~\forall k \in \mathcal{U}_b \label{eq:req-bs-power} \\
    &&& \text{(\ref{eq:dlsinr})}, \text{(\ref{eq:interference})}, \text{(\ref{eq:dlrateuser})}, \text{(\ref{eq:antenna-count})} \quad \forall b \in \mathcal{B}~\forall k \in \mathcal{U}_b
  \end{alignat}
\end{subequations}


\section{Decentralized Precoding Algorithm}\label{sec-anal}
\newcommand{\myl}{\hat{\mathbf{l}}}
\newcommand{\vect}[1]{\boldsymbol{\mathbf{#1}}}
The problem in the form as it stands cannot be solved efficiently with current numerical techniques. This is because 1) the objective function and the constraint are non-convex and non-smooth; 2) the interference calculation requires the central controller to process a very large matrix. In a general non-convex problem it would result in an unacceptable run-time. To deal with these issues, first we transform the original problem in Problem (\ref{eq:prob1}) into the following problem with newly added auxiliary variables:
\begin{subequations}\label{eq:prob2}
\begin{alignat}{2}
   \underset{\mathbf{w}, \bm{\rho}, \hat{\mathbf{P}}, \mathbf{t}, \bm{\xi},\bm{\zeta},\mathbf{s}}{\text{minimize}}~&
  & & -\sum_{b \in \mathcal{B}}c_b \rho_b \\
   \text{s.t.}~&
  & & \sum_{k \in \mathcal{B}} \xi_{bk} \cdot \hat{P}_b^{-1} \geq \rho_b \quad\forall b \label{eq-prob21} \\
  &&& \zeta_{bk} \geq \log_2 (1 + s_{bk}) \geq \xi_{bk}, \quad \forall b,k \label{eq-prob22} \\
  &&& |\mathbf{h}^H_{b,bk} \mathbf{w}_{bk}|^2 / t_{bk} \geq s_{bk}, \quad \forall b,k \label{eq-prob23} \\
  &&& I_{bk} + N_{bk} \leq t_{bk}, \quad \forall b,k \label{eq-prob24} \\
  &&& \sum_{k \in \mathcal{U}_b} \zeta_{bk} \leq \bar{r}_b \quad\forall b \label{eq-prob25}\\
  &&& \theta_b (\sum_{k \in \mathcal{U}_b} ||\mathbf{w}_{bk}||_2^2 + N_b P_{\text{ant}}) + P_{\text{fixed}} + \nonumber\\
  &&& \quad \sum_{k \in \mathcal{U}_b} \zeta_{bk} \cdot P_{\text{SP-Unit}} \leq \hat{P}_b \quad \forall b \label{eq-power}
  \end{alignat}
\end{subequations}

\begin{proposition}
Problem (\ref{eq:prob2}) has the same optimal solution to problem (\ref{eq:prob1}).
\end{proposition}
\begin{proof}
If the constraints with newly added variables are equal at optimum, we can easily check that it is equivalent to any optimum of the Problem (\ref{eq:prob1}). Then by contradiction one can see that these constraints cannot possibly take strict inequality sign at optimum. For example, if Eq. (\ref{eq-prob22}) were to take strict less than, then one can choose a lower value for $\hat{P}_b$ such that equality is assumed, without violating other constraints. Moreover, with smaller $\hat{P}_b$, the lower bound $\rho_b$ can be increased which leads to a lower objective. The same arguments can be made for all newly modified constraints, because any strict inequality would result in ``free lunch'', the adjustments that do not violate existing constraints and improve the objective values. Hence we can say that these constraints must be equal at optimum, at which point it is the same as the problem before transformation. 
\end{proof}
This purpose of this step is to mainly eliminate the non-linear equality constraints, because even when the non-linear functions are convex, equality constraints containing them cannot be convex; hence it is necessary to put them into inequality forms; also this step reduces the amount of coupled terms between the constraints such that it is easier to analyze and solve with existing algorithmic frameworks. This is done by substituting variables with their upper or lower bounds. With the above preposition, one can see that at optimum such inequalities will all turn to equalities. However, the constraints that comes with non-convex, non-smooth terms are still present and needs to be treated in the algorithm development.

\subsection{Smooth Approximation of the 0-norm}\label{subsec-smooth}
The 0-norm counts the number of non-zero elements in a vector. It is difficult to handle in the usual numerical algorithm design procedure because it is not differentiable or even smooth, rendering typical gradient-based methods ineffective. The difficulty does not stop at a lack of gradient: 0-norm constraints can be reduced in polynomial time to the combinatorial optimization problem of selecting the optimum subset, hence it in fact encodes an NP-class problem. 

For optimizing the antenna selection in this manner there are methods, e.g., the typically used branch-and-bound\cite{nau_general_general_brb}. Alternatively, one may use a greedy-based heuristic antenna selection procedure first for picking out a reasonably good subset. Here we propose to use a smoothing approximation to include antenna selection into our algorithmic framework, which incurs a lower complexity and naturally lends its form to gradient-based methods.

The key observation towards making such an approximation is to associate all antennas $j$ in the BS $b$ with random variables $\bm{\zeta} = \{\zeta_{bj}\}_{b \in \mathcal{B}, j \in \{1,2,\cdots, N_b\}}$, which may be called ``switch variables'' and take on values in $[0,1]$. The precoding vector for user $k$ in BS $b$ changes from $\mathbf{w}_{bk} = [w_{bk,1} ~ w_{bk,2}~\cdots~w_{bk,N_b}]^T$, where elements $w_{bk,j}$ are signal weight for one antenna $j$ for BS user $bk$ have sparsity constraints since we want some of them to be zero across all $K$ users, to $[\hat{w}_{bk,1}\zeta_{b1}~\hat{w}_{bk,2}\zeta_{b2}~\cdots~\hat{w}_{bk,N_b}\zeta_{b,N_b}]^T$, where $\hat{w}_{bk,j}$ do not have sparsity constraints. Notice that switch variables $\zeta$ are set per BS antenna and is the same across the users. Now the sparsity requirement need to be expressed in terms of constraints on $\{\zeta_{bj}\}_{b\in\mathcal{B},j \in \mathcal{A}_b}$.

We can still see that the new precoding vector expression is connected to the original problem: if the $\zeta$ random variables are set in a way that take value $0$ with higher probability for some of the antennas more than other antennas, then we can see that it means the sparse solution is to set those precoding values $w_{bk,j}$ to zero. The natural candidate is to model $\zeta$ as binomial random variables, and tweak the success probability as an optimization variable; but its discrete nature makes it hard to work with gradient methods, so instead we consider its continuous approximation \textit{concrete distribution} \cite{maddison_concrete_2016}. This approximation is a parametric family of continuous distributions, which are crafted to allow gradients to be derived at points that correspond to its discrete counterparts.

This approximate distribution which gives the probability distribution on $[0, 1]$ is generated like this, with subscript $b$ for individual base station omitted:
\begin{subequations}
\begin{alignat}{2}
\zeta &= \min(1, \max(0,\eta)) \label{eq:gate-expr}\\
\eta &\sim q(\frac{\eta-\eta_0}{\eta_1-\eta_0};\bm{\phi} ) \label{eq:gate-latent-dist} \\
q(x;\bm{\phi}) &= \frac{\beta \alpha x^{-\beta-1} (1-x)^{-\beta-1} }{ (\alpha x^{-\beta} + (1-x)^{-\beta})^2 }, \quad \bm{\phi} &= \{\alpha,\beta\} \label{eq:dist-model}
\end{alignat}
\end{subequations}

Starting from a uniform random variable from $[\eta_0, \eta_1]$, we first normalize it to $[0, 1]$ as shown in Eq. (\ref{eq:gate-latent-dist}). Eq. (\ref{eq:dist-model}) is the underlying distribution model used, which is a continuous approximated version of binary Bernoulli distribution taking a number in $[0,1]$ as the input, illustrated in Fig.(\ref{fig:concrete}). The ``shape", or how much probability is assigned to the value at two end points is controlled by parameters $\bm{\phi}$ which are $\alpha$ and $\beta$. This random variable is further limited by a hard sigmoid to produce values in $[0,1]$. By experiment we find that it is better to have the starting neighborhood $[\eta_0, \eta_1]$ larger than $[0,1]$, e.g. $[-0.2,1.1]$. 

As a result, the number of active antennas in the precoding vector is approximated by the expectation of the number of non-zero switch values, which can be calculated from the non-zero probability of the distribution given by Eq. (\ref{eq:gate-latent-dist}):
\begin{subequations}\label{eq:Nbapprox}
\begin{alignat}{2}
N_b &\approx \text{E}[\sum_{j\in \mathcal{N}_b}\zeta_{bj}] = \sum_{j \in \mathcal{N}_b} \big( 1 - Q(0; \bm{\phi}_{b,j}) \big) \\
Q(x;\bm{\phi}) &= Q_0(\frac{\eta-\eta_0}{\eta_1 - \eta_0}; \bm{\phi})\\
Q_0(x;\phi) &= \exp \Big( \beta\big(\log x - \log (1-x)\big) - \log\alpha \Big) / \nonumber \\
&\Big(\exp \Big( \beta\big(\log x - \log (1-x)\big) - \log\alpha \Big) + 1\Big)
\end{alignat}
\end{subequations}
where $Q(x;\bm{\phi})$ is the CDF of random variable $\zeta$ in Eq.(\ref{eq:gate-expr}).

In this way, instead of directly optimizing precoding vectors $\mathbf{w}$ for sparse solutions, the distribution parameters $\bm{\phi}_{b,j}$ are optimized such that the induced distribution would push switch variables to $0$ whenever possible. The price to pay is that the expressions involving the precoding vector, e.g., $g(\mathbf{w})$ needs to be replaced by its expectation over the distribution of $\zeta$, which can be approximated by its sample mean: $\frac{1}{S} \sum_{s=1}^S g( \hat{\mathbf{w}} \circ \mathbf{\zeta}^{(s)})$, where $S$ is the number of samples of $\zeta$ and $\zeta^{(s)}$ the $s$-th sample of $\zeta$.
\begin{figure}[h]
\includegraphics[width=0.3\textwidth]{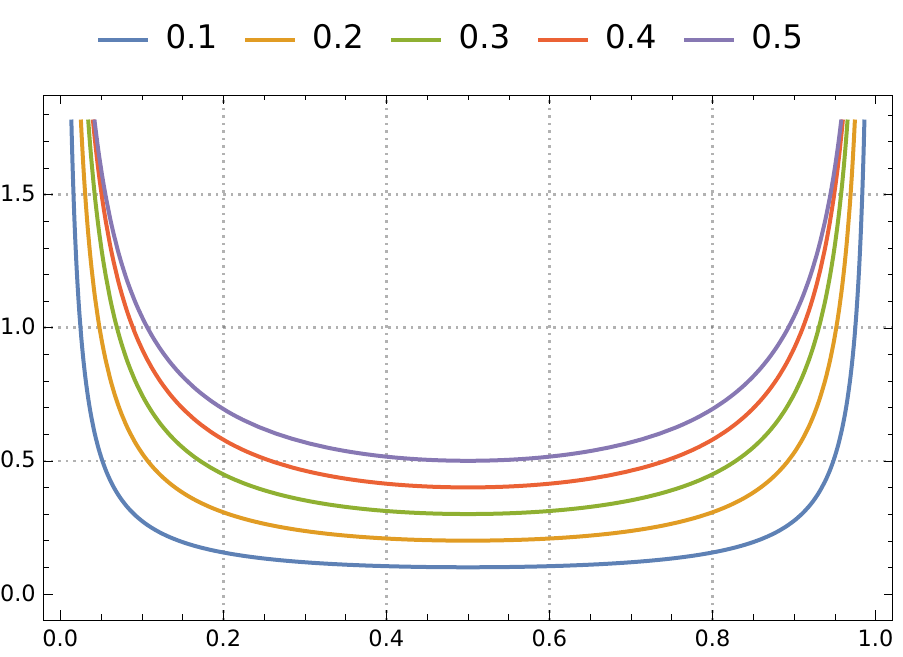}
\caption{The Concrete Distribution when $\alpha = 1$ and $\beta$ takes different values. Notice the similarity with Bernoulli distribution when $\beta$ is close to 0.}
\label{fig:concrete}
\end{figure}

\subsection{Formulation for Decentralization}
With the sparsity constraints smoothed by the above approximation, the next goal is to obtain a proper formulation where the possible separation is exploited to arrive at a decentralized energy efficiency algorithm. The assumption is that through a central controller all the participating BSs could exchange at least partial information on their neighbor's channel so that cooperation is possible, therefore there is a fundamental trade-off between the effects of coordination and the communication overhead; this assumption has been used in most of the literature concerning collaborating cells \cite{lakshminarayana2015coordinated,nauryzbayev2016enhanced}.


This step of transformation uses the consensus form programming \cite{nedic_constrained_2010}. The motivation is simple: constraints are said to be \emph{coupled} if the same variables are used by multiple BSs. Instead of considering them together, it is possible to distribute them to local estimated versions and align them together through iterations. This approach could resolve the inter-dependence between base stations and make it possible to offload part of the computation to the BS side. In this transformation, the cross-interference power terms are treated as the consensus value. We can let $\iota_{b\prime bk}$ be the ICI for cell $b\prime$ to a user $bk$, $\kappa_{b,bk}$ to be the intra-cell interference:
\begin{equation}
I_{bk} = 
\sum_{b\prime \neq b}\sum_{k \in \mathcal{U}_{b\prime}} |\mathbf{h}^H_{b\prime,b k} \mathbf{w}_{b\prime k\prime}|^2 +
\sum_{k\prime \neq k} |\mathbf{h}^H_{b,b k} \mathbf{w}_{b k\prime}|^2
= \sum_{b\prime \neq b} \iota_{b\prime bk} + \kappa_{b,bk}
\end{equation}

Each BS knows their ICI imposed on other, hence at any moment of the system there must exists such a vector that captures all the ICI's, which we define as 
\begin{equation}
\bm{\iota} = \{\bm{\iota}_b\}_b
= \{\iota_{b,b\prime k\prime}\}_{b,b\prime,k\prime} \in \mathbb{C}^{N_c(N_c - 1) \times K}
\end{equation}, because they know the CSI to users outside of their cell. At the same time each BS does not know the ICI caused by others to the users in its own cell, for which it does not have the CSI to calculate, hence there is a need to communicate at a central node. We define
\begin{equation}
\bm{\nu}_{b} = \{ \sum_{b\prime\neq b} \iota_{b\prime,b1}, \sum_{b\prime\neq b} \iota_{b\prime,b2}, \cdots, \sum_{b\prime\neq b} \iota_{b\prime,bK}, \bm{\iota}_{b}\}\in \mathbb{R}^{N_c\times K}
\end{equation}
And it's not hard to see that there is a linear mapping matrix $A_b$ with elements either $0$ or $1$ such that
\begin{equation}\label{eq:consistency}
A_b\bm{\iota} = \bm{\nu}_b
\end{equation}. This form greatly facilitates algorithm development, because $\bm{\nu}_b$ are the ``local'' information which can be calculated at each BS, and when they are gathered together at the center they must be reconciled to satisfy the consistency constraint. Instead of communicating all CSI information, now each BS only transmits $(N_c - 1)\times K$ power information, which does not scale with the massive number of antennas in our problem, and is usually quite small.

With the above transformations, we have smoothed and decoupled the constraints as much as possible; all but one constraint are functions of individual BS's own state $\mathbf{l}_b$. For clarity of notation and the following analysis, we rewrite the problem:
\begin{equation}\label{eq-packing}
  \begin{aligned}
    \mathcal{L} &= \{\mathbf{l}_{b} | ( \text{Constraints} (8b)-(8g),(\ref{eq:Nbapprox})~\text{hold}) \} \\
    \mathbf{l}_{b} &= [ \mathbf{w}_b, \rho_b, \hat{P}_b, \mathbf{t}_b, \bm{\xi}_b,\bm{\zeta}_b,\mathbf{s}_b, \bm{\nu}_b,\bm{\phi}_b ]
  \end{aligned}
\end{equation}
where $\mathcal{L}_b$ is the feasible region of all the BS-local variables. Each BS stores the state information $\mathbf{l}_{b}$, which will be updated in the following iterative algorithm, containing the precoding vectors, ICI, power and EE. The global state of the interference terms are put in vector $\mathbf{\iota}$. The shorthand of the problem is then
\begin{subequations}\label{eq-prob4}
  \begin{alignat}{3}
  \underset{\{\mathbf{l}_b\}_b}{\text{minimize}}~&
  & & { U(\mathbf{l}_b)} \nonumber\\
  \text{s.t.}~&
  & & h_{b}(\mathbf{l}_{b}) \leq 0~\forall b \label{eq:prob41} \\
  &&& A_b\bm{\iota} = \bm{\nu}_b~\forall b \label{eq:prob42}
\end{alignat}
\end{subequations}
where $h_{b}$ is a vector-valued indicator function whenever $\mathbf{l}_{b} \in \mathcal{L},~\forall b$, and constraint Eq. (\ref{eq:prob42}) encapsulates coupled connection between the BSs.

\subsection{Decentralized Precoding Algorithm}
At this point it is tempting to try applying the popular Alternating Direction Multiplier Method (ADMM) like stated in \cite{shen_distributed_2012}. However given the non-convexity of the problem, the direct application of ADMM could result in non-convergence. To remedy this issue a widely adopted approach is successive convex approximation (SCA) as reported in \cite{zhang_energy_2018}, which adds another outer loop outside of the ADMM iterations, causing the run time complexity to be quite high. Here we propose a novel algorithm to address these challenges. It is a combination of both multiplier and quadratic programming(QP) methods, which could solve non-convex problem to a numerical local solution at a reasonably fast converge speed and provable convergence. Note that with non-convex problems it is generally NP-hard to ensure global optimality, hence in engineering problems KKT-condition solutions are considered good enough. 

Based on Problem (\ref{eq-prob4}) which is equivalent to the original problem, we first ensure its feasibility by assuming the feasibility of the original problem in Eq. (\ref{eq:prob1}), and that the local optimum values are KKT points. The input of this algorithm is the utility functions, each BS's local state information in $\mathbf{l}_{b}$, and all BSs have knowledge of the coupling matrix $A$. The output is a KKT point of the original non-convex problem. During the iterations, the notations $\kappa_b$ and $\lambda$ is used to denote the dual variables associated with the non-convex constraints and linear equality constraints, respectively. Notice that we also assume that low-level solvers for convex and quadratic programming are available to a reasonable degree of accuracy.

\paragraph{Initialize} Each BS $b$ needs to give their estimates of their state variables $\{\mathbf{l}_{b}\}_{b=1,2,\cdots,L}$ and the corresponding dual variables $\bm{\lambda}_{b}$ as the initial solution. This initialization can be generated from past channel observations and other BS behavior and the quality should be affect the convergence speed as will shown below. The algorithm uses a few other algorithmic parameters, like the scaling coefficient sufficiently large $\rho > 0$, the weighing positive definite matrix $\Sigma_{b} \in \mathbb{S}_{+}^{n_l}$ and acceptable accuracy $\epsilon$. $\Sigma_{b}$ is a parameter for the generalized vector distance and is set to the identity matrix here.

\paragraph{Update individual estimate of local state} Each BS evaluates the following problem:
\begin{equation}\label{eq-alg2}
  \begin{aligned}
    & \underset{\myl_b}{\text{minimize}}
    & &{-U_i(\myl_b) + \lambda^T A_b \myl_b + \frac{\rho}{2} ||\myl_b - \mathbf{l}_b^{(k)}||^2_{\Sigma_{b}}} \\
    & \text{subject to}
    & &h_b(\myl_b) \leq 0
  \end{aligned}
\end{equation}
In this subproblem, the objective function is of the Augmented Lagrange multiplier form. Note that this is solved in a distributed manner; each BS optimizes what is best for its own objective, expressed in $\myl_b$, given the previous iterate solution $\mathbf{l}_b^{(k)}$. This step generates a first estimate $\myl_{b}$, which together with other calculated results would be part of the new solution in the next iteration. This is similar to the ADMM algorithm, where an Augmented Lagrangian is decomposed with respect to the individual BSs.

\paragraph{Termination condition} The next step checks the termination condition: when the affine constraints are satisfied \emph{and} the individual solutions are sufficiently close to the previous solution, the solution terminates and the current solution is treated as the final solution. This step is to avoid the more expensive negotiation steps when the current solution is good enough. Specifically, the following conditions must be met:
\begin{subequations}\label{eq-algter}
  \begin{alignat}{2}
    & \rho|| \sum_{b} \myl_b - \mathbf{l}_b ||_2^2 \leq \epsilon \\
    & ||\sum_{b\in \mathcal{B}_{C}} A_{b} \myl_{b} - \mathbf{a}||_2^2 \leq \epsilon
  \end{alignat}
\end{subequations}
The output of this algorithm is then set to be $\myl_b$ at this point, which by these two conditions automatically satisfy the KKT conditions of the original problem.

\paragraph{Negotiation} If the above conditions do not hold, this means that the individual solutions of the BSs are not feasible and must be modified, as represented by the violation of the affine constraints. This can be compared to a price negotiation analogy used in \cite{nemirovski_five_2002}: each agent first determine a local version of bidding, then if they do not agree use pricing as indicator to fix their bidding until a consensus is reached. Here at central controller this quadratic programming problem is solved to find a correction value $\Delta\mathbf{l}_b$:
\begin{subequations}\label{eq-alg3}
  \begin{alignat}{3}
    \underset{\Delta \mathbf{l}, s}{\text{minimize}}~&
    & &\sum_{b} (\frac{1}{2} \Delta \myl_b^T \Gamma_b \Delta\myl_b + \tilde{\nabla}U_i(\myl_b) \Delta\myl_b) + \lambda^T s + \frac{\rho_2}{2} ||s||^2_2 \label{eq-alg32} \\
    \text{s.t.}~&
    & & \sum_{b} A_b (\myl_b + \Delta \myl_b) = \mathbf{a} + \mathbf{s} \label{eq-alg33} \\
    &&& C_{b} \Delta \myl_b = 0 \quad \anybs \label{eq-alg34}
  \end{alignat}
\end{subequations}
where
\begin{subequations}
  \begin{align}
    \Gamma_{b} &\approx \nabla^2 \{ -U_b(\myl_b) + \kappa_{b}^T h_{b}(\myl_b) \} \label{eq-alg35} \\
    \tilde{\nabla} &= \nabla f_i(\myl_b) + (J_b - \tilde{J}_b)^T \kappa_i \label{eq-alg36} \\
    J_{i,j} &= \frac{\delta}{\delta x} (h_{i}(x)))_{j} |_{x=\myl_{i}} \text{ if } (h_i(y_i)) = 0 \label{eq-alg37}
  \end{align}
\end{subequations}
In this step the communication between different BSs take place and the computation is delegated to the quadratic programming solver at the central controller. It is motivated from an application of sequential quadratic programming (SQP), a common technique for solving non-convex constrained programming. The intuition is similar to the \emph{trust-region}: it is easier to approximate the original non-linear function with its second-order expansion and optimize this quadratic function instead, within a local region around the current iterate solution. To see how it plays into our algorithm development, first consider the minimization of an augmented Lagrangian function which incorporates both the primal and dual variable, and penalty terms in the first and second order form \cite{rockafellar1973multiplier}:
\begin{equation}\label{eq-lagrangian}
  \mathcal{L}(\myl, \lambda; \rho_2) = - \sum_{b} U_{b} + \lambda^T (\sum_{b} A_{b} \myl_b - a) + \frac{\rho_2}{2} || \sum_{b} A_{b}\myl_b - a||_2^2
\end{equation}
for which we need to find the best directions of $\Delta \myl$ such that $\mathcal{L}(\myl, \lambda)$ is minimized, subject to $h_b(\myl) \leq 0 \quad \forall b$. The new parameter $\rho_2$ is for regulating term and is assumed to be decreased in a way that does not exceed $O(||\mathbf{l} - \mathbf{l}^*||)$, which is common setting for numerical stability. From the theory of quadratic programming we know that there is a direct equivalence between a solving a Lagrangian minimization with Newton method and a quadratic programming \cite{wright1999numerical}, the minimization of Eq. (\ref{eq-lagrangian}) can be approximated by solving a quadratic programming problem:
\begin{subequations}
  \begin{alignat}{3}
    \underset{\Delta \myl}{\text{minimize}}~& 
    & &\sum_{i} \mathcal{L}(\myl + \Delta \myl, \lambda; \rho_2) \\
    \text{s.t.}~&
    & & h_{b}(\myl_b) + C_{b}\Delta \myl_{b} \leq 0 \quad \anybs
  \end{alignat}
\end{subequations}
where $C_{b}$ is the Jacobian for all the non-convex constraints $h_b$. When we expand this into this problem's variables, will come in the same form as the step in Eq. (\ref{eq-alg34}).
\begin{subequations}\label{eq-algqp}
  \begin{alignat}{3}
    & \underset{\Delta \myl, s}{\text{minimize}}
    & &  \sum_{b} \frac{1}{2} \Delta \myl_b^T \bm{H}_{b} \Delta \myl_b + g_{b}^T \Delta \myl_b + \mathbf{\lambda}^T\mathbf{s} + \frac{\rho_2}{2} ||\mathbf{s}||^2_2\\
    & \text{subject to}
    & & h_{b}(\myl_b) +  C_{b}\Delta y_{b} \leq 0 \quad \anybs\\
    &&& \sum_{b} A_b(\myl_b + \Delta\myl_b) = \mathbf{b} + \mathbf{s} \quad \anybs
  \end{alignat}
\end{subequations}

For numerical stability reasons another slack variable $s$ is added. Solving this subproblem would essentially give a correction to the local solutions gained from step b), which would help with the violation of constraints as expressed in the stopping conditions Step 2) in Algorithm \ref{alg:xxx}. The effect of the usage of quadratic programming is two-fold: the superior quadratic convergence helps reduce the overhead in the correction step, and the use of second-order term enables simpler analysis of the convergence. As the problem we aim to solve involves the inter-dependency between BSs, the price to pay is reflected here: each BS need to maintain the interference level caused by others and by itself, therefore the total problem scale in this step is $O(KL^2)$, where $L$ and $K$ are the numbers of cells and users per cell, respectively. In our examples this is still feasible because the number of cooperating cells $L$ does not grow to a very large value, since the limiting small-cell BS hardware capabilities do not allow cooperations of a large group. Although this is an expensive step, it is amenable to standard techniques which are implemented in available solvers like SeDuMi or SNOPT.

To further reduce the computation cost at this step, approximations can be used. In the proposed algorithm,  inaccurate Jacobian update rules \cite{wright1999numerical} are used to save space and time. Also $\Gamma_b$ can also be the approximate version of the Hessian matrix, which can calculated with only first-order gradient information like reported in \cite{byrd_limited_1995}.

\paragraph{Update Solution}
The new iteration solution is given by
\begin{equation}\label{eq-algupdate}
  \begin{aligned}
    \mathbf{l}_b^{(k+1)} &= \myl_b - \Delta\myl_b \\
    \bm{\lambda}^{(k+1)} &= \bm{\lambda}
  \end{aligned}
\end{equation}

This is the final step where the calculated result is used to update the current solution. The final step does not require additional communication from other BSs and is done at local level.

\subsection{Optimality and Convergence Property}
Assuming the final utility function is twice continuously differentiable, this algorithm then can be shown to be locally convergent. This means that if the estimated solution is sufficiently close to the optimum solution, then the algorithm will converge. This is because the step in Eq. (\ref{eq-alg34}) follows from the convergence results of SQP methods \cite{wright1999numerical}. We first show that the optimum result in the distributed optimizations in step a) of the algorithm provides a reasonably good solution:

\begin{lemma}\label{eq-prop1}
  Given twice continuously differentiable utility functions $U_b(\mathbf{l})$, and the KKT point $(\mathbf{l}^*, \lambda^*)$, and the condition $\nabla^2 _{} (U_{\mathbf{l}}(\mathbf{l}_b^*) + \kappa_{b}^T h_{b}(\mathbf{l}_b^*)) + \rho \Sigma_{i} \succ 0$ holds for some $\rho > 0$, and that $\mathbf{l}$ and $\mathbf{l}^*$ are sufficiently close, then the step in Eq. (\ref{eq-alg32}) has locally unique minimizers $\myl$ such that there exist constant $k_1 > 0$, $k_2 > 0$
  \begin{equation}
    ||\myl - \mathbf{l}^*|| \leq k_1 ||\mathbf{l} - \mathbf{l}^*|| + k_2 || \lambda - \lambda^*||
  \end{equation}
\end{lemma}

\begin{proof}
  First recall that
  \[
    \myl(\mathbf{l}^{(k)}, \lambda^{(k)}) = \argmin_{\myl_b} L = \argmin_{\myl_b} U_b(\myl_b) + \lambda^T A_{b}\myl_b + \frac{\rho}{2} || \myl_b - \mathbf{l}_b^{(k)} ||^2_{\Sigma_b}
  \]
  It is easy to check that the Hessian of the Lagrangian in Eq. (\ref{eq-alg32}), which is
  \begin{equation}
    \nabla^2 (-U_i(\mathbf{l}_b^*) + \kappa_{b}^T h_{b}(\mathbf{l}_b^*)) + \rho\Sigma_{b}
  \end{equation}
  are all positive definite for all $(\mathbf{l},\lambda)$ sufficiently close to the optimum $(\mathbf{l}^*, \lambda^*)$. Then the minimization results in Eq. (\ref{eq-alg32}) are well defined and differentiable in this neighborhood. The statement then holds because $|| \frac{\partial L}{\partial x}|| < k_1$ and $||\frac{\partial L}{\partial \lambda}|| < k_2$ hold from being evaluated at KKT points.
\end{proof}

We could then show that the distance between next iterate solutions $\mathbf{l}^{k+1}$ and $\lambda^{k+1}$ and their KKT points $\mathbf{l}^*$ and $\lambda^*$ are both upper-bounded:
\begin{theorem}\label{eq-prop2}
  If the exact Hessian and Jacobian are used in the step Eq. $(\ref{eq-alg34})$, then there exists a constant $\rho$ such that
  \begin{equation}
    \begin{aligned}
      & ||\mathbf{l}^{(k+1)} - \mathbf{l}^*|| \leq \frac{\rho}{2} ||\myl - \mathbf{l}^*||^2 \\
      & ||\lambda^{(k+1)} - \lambda^{*}|| \leq \frac{\rho}{2} ||\myl - \mathbf{l}^*||^2
    \end{aligned}
  \end{equation}
\end{theorem}

\begin{proof}
The quadratic programming problem in Eq. (\ref{eq-alg3}), if at optimum, must satisfy the first optimality test. Specifically by using Newton equation in solving for the best Lagrange direction, we have such equality:
\[
\begin{bmatrix}
H^* &  A^T & C \\
A   & \frac{1}{\rho_2} I & 0\\
C   & 0  &  0
\end{bmatrix}
\begin{bmatrix}
\Delta \myl\\
\lambda^{(k+1)} - \lambda^{(k)}\\
\kappa_{QP}
\end{bmatrix}
=
\begin{bmatrix}
\nabla_y \sum_{b} -U_i + A^T\lambda\\
a - A\myl\\
0
\end{bmatrix}
\]

which is equivalent to applying the Newton method for solving the original problem in Eq. (\ref{eq-alg3}). Given that exact Hessian and Jacobian matrices are sued here, the source of inaccuracy comes from the addition of $\frac{1}{\rho_2}I$, which goes to zero if $\rho_2$ is sufficiently large as solution is approached, see the discussion about Eq. (\ref{eq-alg3}). Hence the iteration has the same quadratic convergence as the Newton method does. Namely,
\begin{equation}
  \begin{aligned}
    &||\mathbf{l}^{(k+1)} - \mathbf{l}^{*}|| \leq \alpha ||\myl - \mathbf{l}^{*}||^2\\
    &||\lambda^{(k+1)} - \lambda^{*}|| \leq \alpha ||\myl - \mathbf{l}^{*}||^2
  \end{aligned}
\end{equation}
holds for some $\alpha > 0$.

Combine this result with the above result from Lemma (\ref{eq-prop1}), it is not hard to arrive at
\begin{align}
  &k_1|| \mathbf{l}^{(k+1)} - \mathbf{l}^{*} || + k_2 ||\lambda^{(k+1)} - \lambda^{(k)}|| \\
  \leq &\frac{\alpha k_1 + \alpha k_2}{2} ( ||\mathbf{l}^{(k+1)} - \mathbf{l}^{*}||^2
+ ||\lambda^{+} - \lambda^*||^2) \nonumber
\end{align}
given that $k_1$, $k_2$ are all positive coefficients, this establishes the local quadratic convergence.
\end{proof}


%
\begin{algorithm}[htbp]

 \caption{Proposed EE Optimization Algorithm}\label{alg:xxx}

\textbf{Input:} Initial guesses $\mathbf{l}_b$ and dual variables $\lambda$, $\kappa_b$ and a numerical tolerance $\epsilon > 0$, and algorithm parameters $\rho, \rho_2$

\textbf{Do:}

\begin{enumerate}
\item Choose a sufficiently large penalty parameter $\rho \ge 0$ and positive semidefinite scaling matrices $\Sigma_i \in \mathbb{S}_{+}^{n_x}$ and solve for all $i \in \{1,\dots ,N\}$ the decoupled optimization problems Prob. (\ref{eq-alg2}) to a KKT solution.

\item If $\parallel \sum_{b} A_b \myl_b - a \parallel \le \epsilon$ and $\rho\parallel \sum_{b}(\myl_b - x_b) \parallel_1 \le \epsilon$, terminate with $\mathbf{l}_b^* = \myl_b$ as the solution.
\item Choose approximations $C_i \approx C_i^*$ of the Jacobian matrices $C_i^*$ defined by
    \[ C_{i,j}^* =
  \begin{cases}
    \left. \frac{\partial}{\partial x}(h_i(\mathbf{l}))_i \right |_{x=\myl_b} & \quad \text{if } (h_i(\myl_j))_j = 0 \\
    0  & \quad \text{otherwise}\\
  \end{cases}
\]
    Compute the modified gradient
    \[
    g_b = \nabla -U_b(y_b) + (C_b^*-C_b)^T \kappa_i\]
    and calculate Hessian approximations
    \[
    H_i \approx \nabla^2 \{-U_b(\myl_b) + \kappa^T h_b(\myl_b)\}\].
\item Choose a sufficiently large penalty parameter $\rho_2 > 0$ and solve the coupled QP in Eq. \ref{eq-algqp}.
\item Update the iterates as in Eq.(\ref{eq-algupdate}) and continue with step 1.
\end{enumerate}
\end{algorithm}

\subsection{Parameter Tuning for Global Convergence}
This algorithm can be shown to be convergent regardless of the initial solution quality, if further steps are taken in the choice of step parameters $\alpha_1, \alpha_2,$ and $\alpha_3$. To show that the convergence depends on these values, consider the penalty function and the difference between successive steps. We can define the penalty function in terms of how severe the constraints are violated and how optimal the objective function are:
\begin{equation}
  P(\mathbf{l}_{b}) = \sum_{b} -U_{b}(\mathbf{l}_{b}) + \kappa_L \sum_{b} \max(h_{b}(\mathbf{l}_b), 0) + \lambda_L\sum_{b} ||A_b l_b - a ||
\end{equation}
where $\kappa_L$ and $\lambda_L$ are sufficiently large positive constants. The difference $\Delta P = P(l^{(k+1)}) - P(l^{(k)})$, if can be shown to be lower bounded by a constant above zero, then can lead to the convergence of this algorithm. However, due to non-convexity, it is possible for the algorithms to be stuck at points where the primal variables $\mathbf{l}_b$ changes but $|\Delta P|$ stagnates. Therefore new strategies for fine-tuning the algorithm is needed.

\begin{lemma}\label{lemma:suff-lowerbound}
If the difference of successive penalty function is lower bounded by
\begin{equation}\label{eq:suff-lowerbound}
\gamma ( \sum_{i=1}^N ( \frac{\rho}{2} ||y_i-x_i||_{\Sigma_i}^2 + \lambda ||\sum_{i=1}^N A_iy_i - b||_1 ) )
\end{equation} then it is sufficient that the algorithm always terminates after a finite number of iterations.
\end{lemma}
\begin{proof}
By considering the contrapositive, i.d., if the algorithm cannot terminate, meaning that either one of the loop-breaking conditions listed in step 2) are not met. This would suggest that either one of the following is true: 
\begin{subequations}
\begin{align}
&||\sum_{i=1}^N A_iy_i - b||_1 > \epsilon\\
&\sum_{i=1}^N \rho || y_i - x_i ||_{\Sigma_i}^2 > \epsilon
\end{align}
\end{subequations}

Therefore expression (\ref{eq:suff-lowerbound}) in either case must be at least as large as one of the two, implying
\begin{align}
\Delta P & \geq \gamma ( \sum_{i=1}^N ( \frac{\rho}{2} ||y_i-x_i||_{\Sigma_i}^2 + \lambda ||\sum_{i=1}^N A_iy_i - b||_1 ) )\\
         & \geq \gamma \min (\lambda_L\epsilon, \frac{1}{2\rho}\epsilon^2)
\end{align}
which is a positive constant and a function of the algorithm constants $\epsilon$, $\lambda$ and $\rho$ only. When the original problem is feasible, the penalty function must have a minimum. With each successive step it decreases by at least such a positive amount, it is certain that the algorithm can exit after finitely many iterations.
\end{proof}

By this lemma we need only to consider the cases when the condition (\ref{eq:suff-lowerbound}) are not satisfied. To see how this is entirely possible, suppose that the initial solution is very far from the optimum. In this case with large penalty terms the difference between successive penalty function values can be small enough to fail the condition check. To provide a backup solution iterate $\tilde{x}^{(k+1)}$ which can always leave (\ref{eq:suff-lowerbound}) satisfied, one can consider such an auxiliary problem:
\begin{equation}\label{eqs:aux1}
  \begin{alignedat}{2}
    \underset{y}{\text{minimize}}~&
    &&  \sum_{i} f_i(y_i) + \frac{\rho}{2} ||y_i - x_i||^2_{\Sigma_{i}} \\
    \text{s.t.}~&
    &&  h_{i}(y_i) \leq 0 \quad \forall i\\
    &&& \sum_{i}^N A_{i}x_i - b = 0 \quad \forall i
  \end{alignedat}
\end{equation}

This problem form comes from the regularization terms commonly used in proximal optimizations \cite{combettes_proximal_2011}. Solving this problem (\ref{eqs:aux1}) yields a next iteration $y^*$ that will satisfies the condition (\ref{eq:suff-lowerbound}).

\begin{lemma}
The solution to the auxiliary problem (\ref{eqs:aux1}) provides an update that will satisfy the sufficient convergence condition.
\end{lemma}

\begin{proof}
Due to its optimality, the difference in $P$ function is strictly positive.
\begin{equation}
\begin{aligned}
    &P(x) - P(y^*)\\
  = &\sum_{i=1}^N (f_i(x_i) + \frac{\rho}{2}|| x_i -x_i ||^2_{\Sigma_i}) + 
    \lambda_L || \sum_{i=1}^N A_ix_i - b || + \kappa_L \sum_i \max \{ 0, h_i(x_i)_j \} \\
                  & - \sum_{i=1}^N (f_i(y_i^*)  - \lambda_L || \sum_{i=1}^N A_iy_i^* - b || - \kappa_L \sum_{i,j} \max \{ 0, h_i(y_i^*)_j \}\\
      \geq & \frac{\rho}{2}\sum_{i=1}^N|| y_i^* -x_i ||^2_{\Sigma_i} \\
      = & \frac{\rho}{2} \sum_{i=1}^N || y_i^* - x_i ||^2_{\Sigma_i} + \lambda_L||\sum_{i=1}^N A_i y_i^* - b||
\end{aligned}
\end{equation}
The last step shows that the difference has a lower bound of the same form as the condition Eq.(\ref{eq:suff-lowerbound}), hence it can ensure a strict lower penalty.
\end{proof}%

Now that with this lemma we can always find a iterate solution than can converge, the natural step to take next is to see if it can be solved with its dual. This way it could better blend in a primal-dual update framework employed in the main algorithm. Fortunately, the duality gap can be shown to be zero between the auxiliary problem and its dual. This is shown by the following lemma.

\begin{lemma}
The auxiliary problem's dual problem
\begin{subequations}\label{eq:probaux}
\begin{alignat}{2}
\underset{\lambda}{\text{maximize}}~&
&&\inf_{y} \sum_{i=1}^N (f_i(y_i) + \lambda^T A_iy_i + \frac{\rho}{2} ||y_i - x_i||^2_{\Sigma_i} - \lambda^T b\\
\text{s.t.}~&
&& h_i(y_i) \le 0 \quad \forall i
\end{alignat}
\end{subequations}
has a zero duality gap.
\end{lemma}

\begin{proof}
This proof follows the outlines given by literature on proximal operator analysis in Theorem 1 of \cite{ruckmann_augmented_2009}.
\end{proof}

In effect, we discovered a sufficient condition on the iterate solution that will ensure the strict decrease in the penalty function, and then construct an auxiliary problem whose solution satisfies it. Additionally, this auxiliary problem has a dual problem where there is no duality gap, making it a natural choice in a primal-dual update iteration. With these additional tuning steps the algorithm can ensure the convergence even when the initial guess is far from the optimum solution. This process is another demonstration of the compromise between convergence and complexity commonly found in the algorithm design. This additional procedure is summarized as below:

\begin{algorithm}[htbp]
\caption{Step Size Tuning Procedure for Convergence Assurance}\label{alg:step}
\SetKwFunction{Fn}{step\_size\_tuning}
\KwData{Previous and current iterate solutions $x^{(k+1)},x^{(k)}$}
\KwResult{$\alpha_1, \alpha_2, \alpha_3$ as used in Algorithm (1)}

\Fn{$x^{(k+1)}$,$x^{(k)}$}{

$\alpha_1 := 1$, $\alpha_2:=1$, $\alpha_3:=1$\;
$\lambda_L \gg 1$, $\kappa_L \gg 1$, $0 < \gamma \ll 1$\;
\If{$\Delta P >$ condition}{
  $\alpha_1 \leftarrow 1$, $\alpha_2 \leftarrow 1$, $\alpha_3 \leftarrow 1$
}
\ElseIf{$\Delta P(y) >$ condition}{
  $\alpha_1 \leftarrow 1$, $\alpha_2 \leftarrow 1$, $\alpha_3 \leftarrow 1$
}
\Else{
  Solve the dual auxiliary problem (\ref{eq:probaux}):

  $\alpha_1 \leftarrow 1$, $\alpha_2 \leftarrow 1$, $\alpha_3 \leftarrow \alpha^*$
 
}
}
\end{algorithm}


\section{Numerical Results}\label{sec-num}
This section illustrates the proposed algorithm's results in the scenario depicted in Fig. \ref{fig-system}. We perform a system-level multi-cell simulation, where each cell covers a square area with each side 200 meters. The BS are put together on a grid, and are assumed to ``wrap around'', i.e., upper most cell is a neighbor of lower most cell, for equal treatment of edge users. These small cell base stations form a cooperative cluster to represent a hotspot coverage tier. We assume that they are connected through either fiber or wireless backhaul and the macro cell BS uses non-overlapping resource blocks. The users are distributed evenly in the cell, with a minimum distance between each other and from the BS. 

The channels under consideration are flat and Rayleigh within a coherence block, i.d. $\mathbf{h}_{b,bk}$ are i.i.d with $ \mathcal{CN}(0,\mathbf{\Theta}_{b,k})$ where $\bm{\Theta}_{b,k}$ is the correlation. The large-scale fading coefficients satisfy log-normal distribution with Urban Macro Model of $-139.5 - 35\log_{10} d_{b,b\prime k\prime} + \Upsilon$ where $\Upsilon$ is the shadow fading with Gaussian distribution \cite{3gpp.25.996}. Unless otherwise noted SBS transmitters have a budget of $40$ dBm. The receiver noise are set to be $N_n = -120$ dBm in a band of $10$ MHz. The other experiment parameters are set in a similar fashion as in \cite{bethanabhotla2014user}.

In the simulations we consider the following approaches for comparison. As a default baseline we use zero-forcing (ZF) precoding and equal power distribution with no BS antenna power optimization; then we compare the effects of incorporating these factors to justify our choice. For algorithm comparison we also compare the convergence behavior with respect to another reported used method, Dinckelbach method for ratio programming.

\subsection{Performance Evaluation}
In Fig. \ref{fig:ee-powerbudget}, we compare the effects of using different precoding schemes on the system energy efficiency. This simulation is done with default settings and vary the available power budget for each BS antenna. We can observe that when the amount of power is low there is not much of a difference in the system EE. This is mainly due to three precoding schemes offering similar level of throughput. When the available power increases to a typical level, the proposed algorithm achieves a higher EE, which is due to the fact that with a larger search space it is easier for the algorithm to make use of antennas and balance it with system power usage. Also the level of possible ICI is also higher, which is not factored in other schemes. However, this leads to performance saturation in high power regime, suggesting that the system EE gain by optimizing precoder is not high.
\begin{figure}[h]
\begin{subfigure}{.45\linewidth}
  \includegraphics[width=\textwidth]{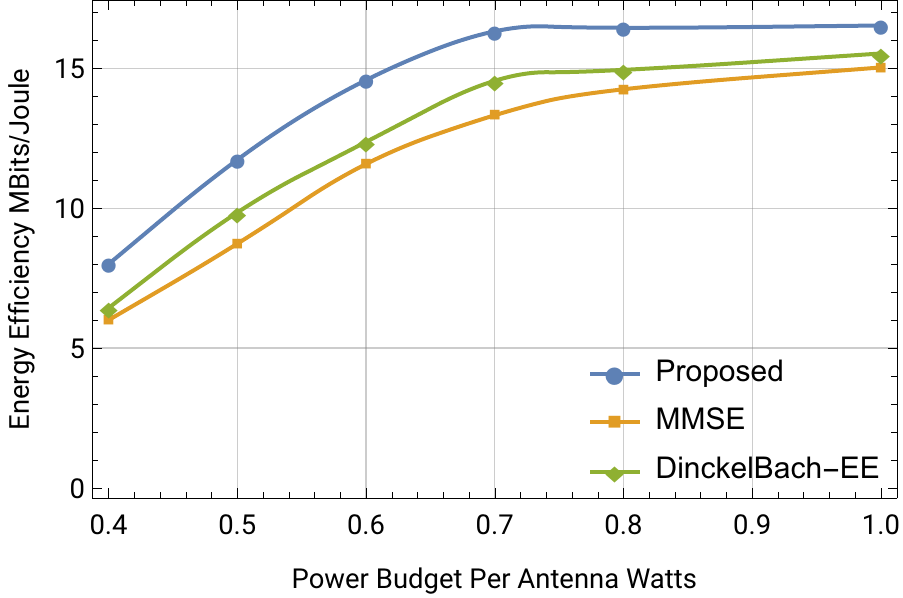}
  \caption{System EE versus Power Budget}
  \label{fig:ee-powerbudget}
  \end{subfigure}
  ~
  \begin{subfigure}{.45\linewidth}
  \includegraphics[width=\textwidth]{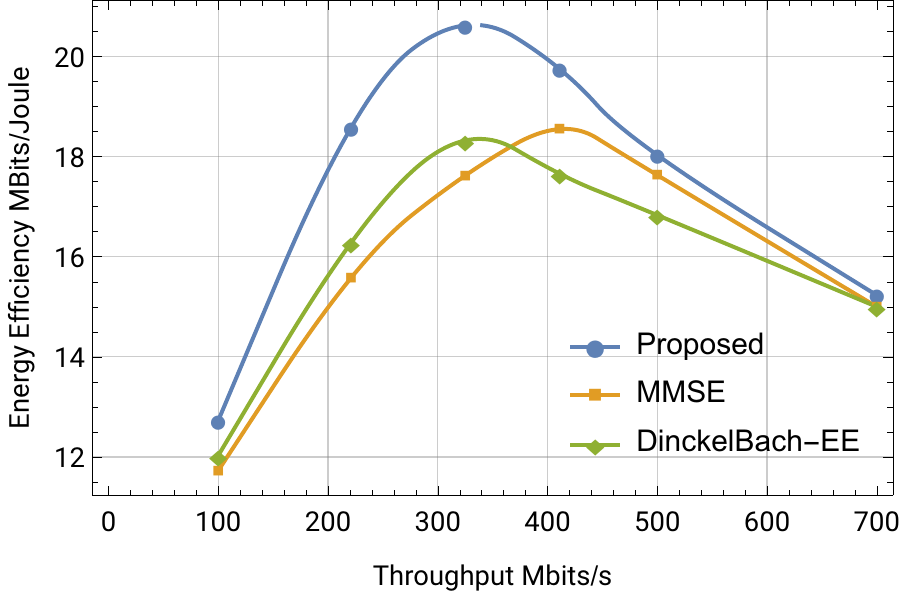}
  \caption{System EE versus SINR requirement}
  \label{fig:ee-sinr}
\end{subfigure}
\end{figure}

In Fig. \ref{fig:ee-sinr}, the relationship of the system EE and system SINR is explored. The result show that different algorithms have a similar trend in the EE-SE relationship: there exists a regime where SE and EE can be improved simultaneously, and after the best ``sweet spot'' the fundamental trade-off between the two starts to appear. Their maximum EE are all achieved at similar levels of system throughput, at around 350 Mbit/s. However, the proposed scheme still achieves a higher EE due to the fact that in the low throughput regimes, opportunistically turning off antennas could result in power saving while achieving similar level of throughput.

In Fig. \ref{fig:ee-bsdensity}, we compare the performance of system EE when the base station density changes. We control this by setting the distance between BS's, and can see from the results that in the case of higher BS density, the proposed algorithm can help find precoding scheme that balance between reducing the inter-cell interference using existing diversity and maintaining high efficiency; the compared schemes have lower efficiency in this scenario because the limited coordination between the BS's.
\begin{figure}[h]
\centering
\begin{subfigure}[b]{0.45\textwidth}
  \includegraphics[width=\textwidth]{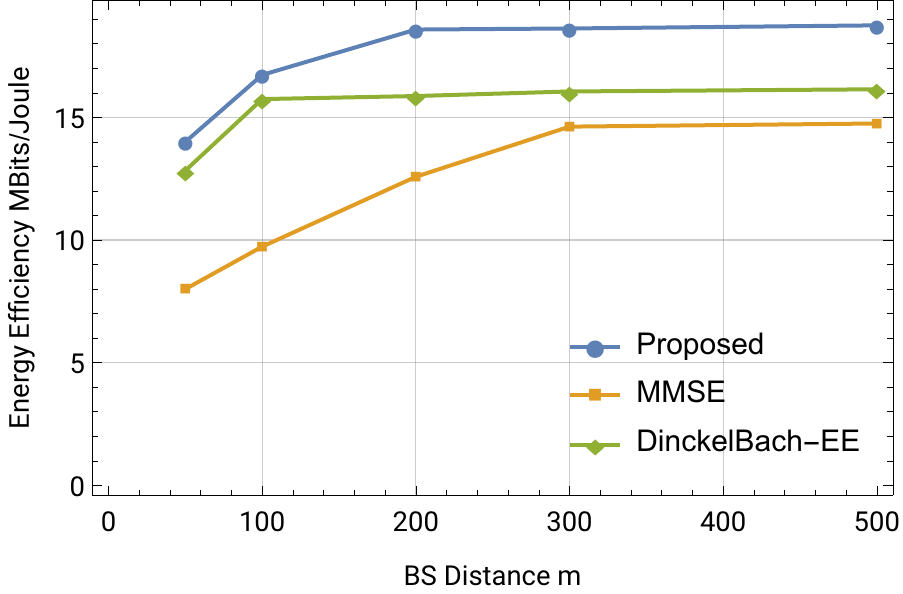}
  \caption{System EE versus BS Distance}
  \label{fig:ee-bsdensity}
  \end{subfigure}
  ~
  \begin{subfigure}[b]{0.45\textwidth}
  \includegraphics[width=\textwidth]{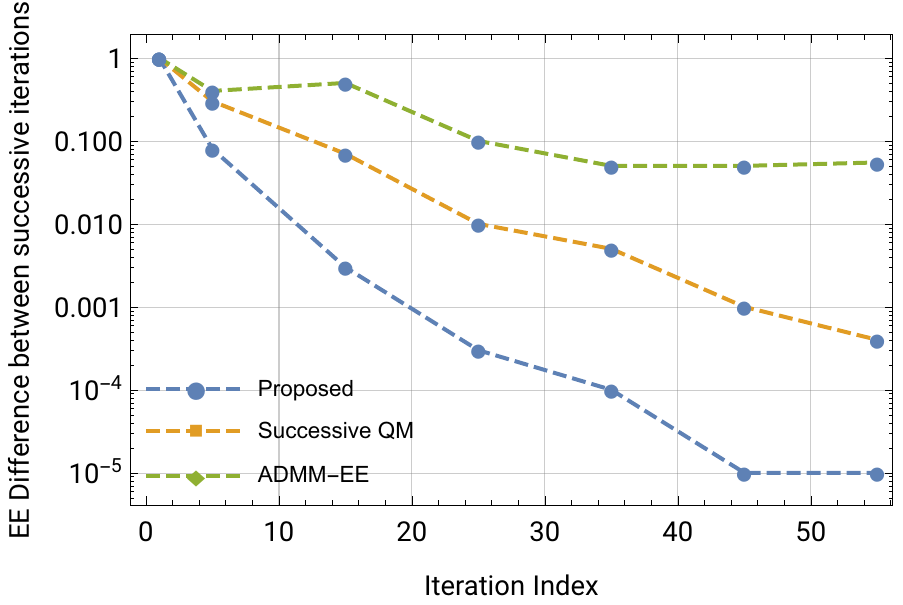}
  \caption{System EE Change versus Iterations}
  \label{fig-ee-t}
\end{subfigure}
\end{figure}

In Fig. \ref{fig-ee-t}, the convergence behavior of the algorithm is given and compared with that from another work \cite{mili_joint_2016}, where a related but different method, successive quadratic programming is used. With equal initial parameters and solution estimates, the superior convergence of proposed scheme is demonstrated by the better descent of the distance with the solution value. Even when they eventually achieve similar linear convergence in the long run, better initial descent in the first rounds is shown to important. As another comparison, the EE objective with direct ADMM methods are used. As can be observed, in this case the algorithm have convergence issue after the first iterations.

\subsection{System Parameter Selection}
In this subsection we are interested in using the proposed algorithm to see how it could help guide a system designer to choose appropriate parameters towards a more energy efficient system. The scenario under consideration is a cooperating SBS cluster, and we apply the algorithm to loop through system parameters $K$ and $N_b$, in the ranges $\{20-100\}$ and $\{50-200\}$ respectively, to look for the optimum system energy efficiency configurations.

We can see in Fig. \ref{fig:ee-3d} that even when we perform an optimum precoding scheme, the large scale system behavior remains similar to the theoretical analysis. With growing user number per cell and larger number of BS antennas, the energy efficiency rapidly increases to a high point, then further increase would drag the system energy efficiency down. The optimum level of BS antennas to UE closely matches the one-order-of-magnitude thumb rule, ranging from $3$ - $8$, depending on the system configurations.

\begin{figure}[t!]
\begin{subfigure}[t]{.45\textwidth}
\centering
\includegraphics[width=\textwidth]{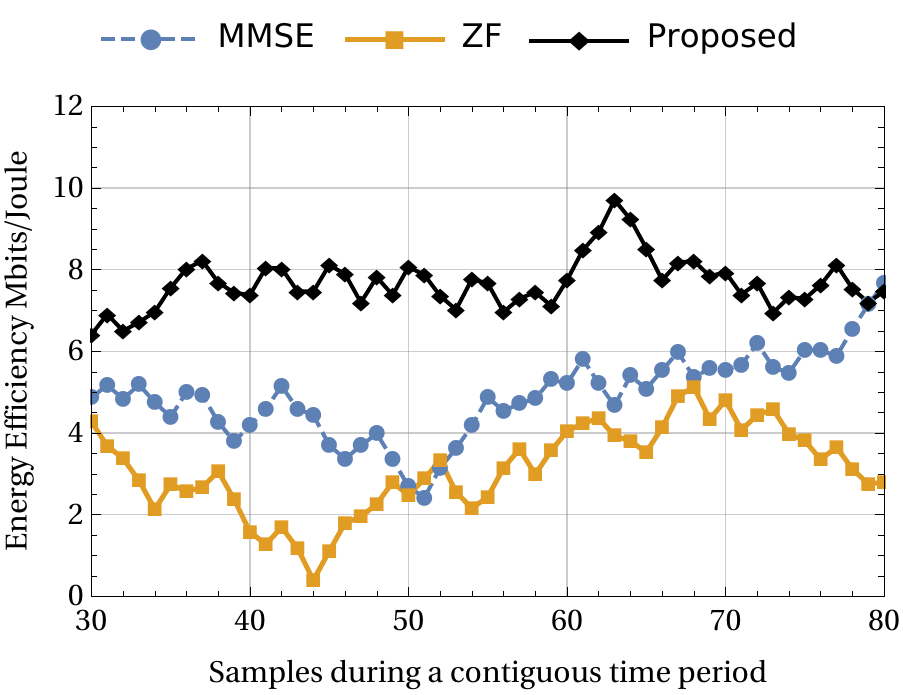}
\caption{System EE versus Time}
\label{fig-ee-time}
\end{subfigure}%
~
\begin{subfigure}[t]{.45\textwidth}
\centering
\includegraphics[width=\textwidth]{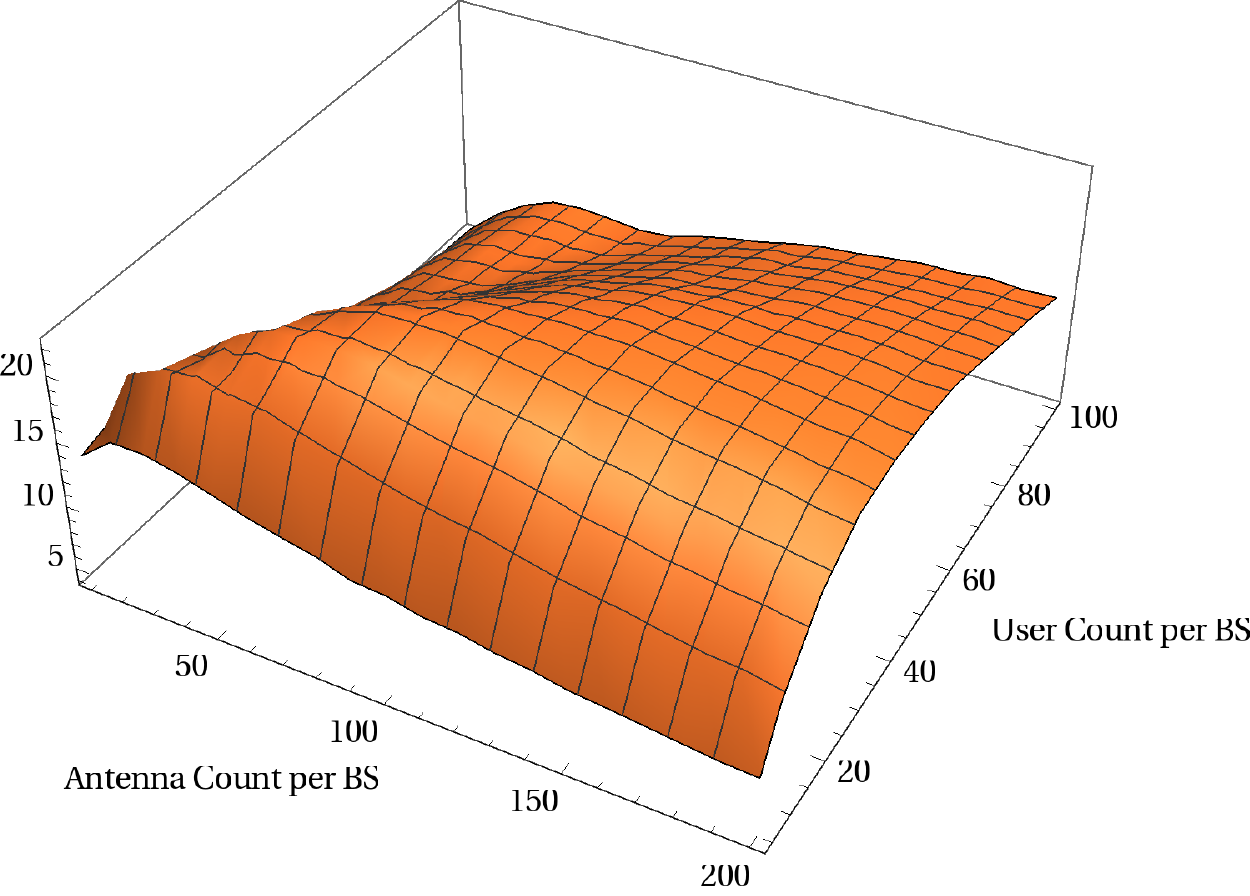}
\caption{System EE versus the average user count per BS and Antenna count per BS}
\label{fig:ee-3d}
\end{subfigure}
\end{figure}

The results suggest that 1) massive MIMO base stations is the way for high energy efficiency communication system, and despite a more refined BS power modeling, its high spectral efficiency is still a dominant factor in its energy efficiency; 2) the optimum EE operating configurations in massive MIMO is sensitive to many parameters like hardware efficiency and signal processing cost, hence system designers need to build a robust and representative model to fully utilize its potential. 3) EE as a design goal is in most cases not optimum in SE or power consumption; it is neither an indicator of Pareto optimality. When it comes to a multi-goal optimization problem like this, energy efficiency can only serve as an estimation of system efficiency.

\section{Summary}\label{sec-summary}
In this paper we consider the energy-efficient control of BS power allocation, switching policy, antenna selection and beamforming in an integrated framework, in the setting of two-tier Massive MIMO HetNet. We formulate this problem as a network utility maximization problem, and propose a decentralized algorithm to solve it, using numerical techniques from sequential quadratic programming and augmented multipliers. The proposed scheme is evaluated in numerical experiments and is demonstrated to achieve a superior performance.

{\renewcommand{\bibfont}{\small} \printbibliography}

\end{document}